\documentclass[prx,letterpaper,aps,superscriptaddress,floatfix,twocolumn,nofootinbib]{revtex4-2}
\pdfoutput=1
\usepackage{amsmath}
\usepackage{amssymb}
\usepackage{amsthm}
\usepackage{amsfonts}
\usepackage{bbm}
\usepackage{comment}
\usepackage[normalem]{ulem}
\usepackage{graphicx}
\usepackage[dvipsnames]{xcolor}
\usepackage{color,framed}
\usepackage{hyperref}
\definecolor{darkgreen}{rgb}{0,0.8,0.2}
\hypersetup{
    colorlinks=true,
    citecolor=darkgreen,
    linkcolor=blue,
    urlcolor=magenta
}
\usepackage{enumerate}
\usepackage{lipsum}
\usepackage{slashed}
\usepackage{url}
\usepackage{bbm}
\usepackage{chngcntr}
\usepackage{mathrsfs}
\counterwithout{equation}{section}

\usepackage{braket}
\usepackage{graphicx}
\usepackage{tikz}
\usetikzlibrary{decorations.pathreplacing,calc}
\definecolor{myred}{HTML}{a82a2a}

\def\eea{\end{eqnarray}}

\def\Tr{ {\rm Tr} }
\def\<{\langle}
\def\>{\rangle}

\def\bZ{\mathbb{Z}}

\usepackage{amsmath}
\usepackage{comment}
\usepackage{amssymb}
\usepackage{amsthm}
\newtheorem{thm}{Theorem}

\newtheorem{lemma}{Lemma}

\theoremstyle{definition}
\newtheorem{defin}{Definition}

\usepackage{color}
\usepackage{tikz-cd}
\usepackage{comment}

\usepackage{tikz}
\usepackage[export]{adjustbox}

\def\[#1\]{%
\begin{equation}\begin{gathered}#1\end{gathered}\end{equation}%
}

\def\Sum_#1{
\sum_{\substack{#1}}
}

\begin{document}

\title{Perturbatively Stable Self-Correcting Classical Memory from Gauge Averaging}
\author{Ryan Thorngren}
\email{ryan.thorngren@physics.ucla.edu}
\affiliation{Mani L. Bhaumik Institute for Theoretical Physics, Department of Physics and Astronomy, University of California, Los Angeles, California 90095, USA}

\begin{abstract}
    We show that the self-correcting memory in 3d Wegner gauge theory is stable to arbitrary small enough perturbations of the Hamiltonian. Our proof relies on a new method we dub ``gauge averaging'', which gives conditions under which explicitly broken gauge symmetries are effectively restored by fluctuations. These conditions show that the self-correcting memory phase is fluctuation stabilized, with its robustness to perturbations increasing with increasing temperature, up to some $T_c > 0$.
\end{abstract}

\maketitle

\textbf{Introduction---}The ability of physical systems to store information reliably for long periods of time is foundational to modern computing technology. In this work, we are interested in classical systems which do so in thermal equilibrium, in particular, without error correcting interventions. Such a system is therefore called self-correcting \cite{Bravyi_2009,landon2013local,brown2016quantum}. An example is the $\mathbb{Z}_2$-symmetric Ising model, which in its low temperature ferromagnetic phase can remain in a macroscopically magnetized state for a time which is exponentially long in the system size \cite{thomas1989bound,martinelli2004lectures}. However, the memory requires the symmetry: if we break the $\mathbb{Z}_2$ symmetry with an external magnetic field, then macroscopic magnetization can be destroyed in constant time, like degaussing a harddrive.

We seek perturbatively stable self-correcting memory, which remains self-correcting under arbitrary small enough perturbations. Such systems cannot have a local order parameter. Instead, they have topological order. The information is stored by a topological winding number, which is stable for an exponentially long lifetime.

The paradigmatic example of this is Wegner's 3d $\mathbb{Z}_2$ gauge theory, \cite{Wegner1971duality}. This model can be thought of as a model of fluctuating surfaces $S$ with energy $|\partial S|$. At low temperatures, typical surfaces can be assigned a homology class, and this homology class is exponentially long lived. However, general perturbative stability of this memory has yet to be rigorously shown.

Indeed, one might reasonably expect that this memory is \emph{not} perturbatively stable. For example, adding surface tension creates an extensive energy penalty for states in non-trivial homology sectors. However, it was demonstrated numerically in \cite{Poulin_2019} and proved in \cite{stahl2026slow} that as long as the surface tension is small enough, the memory survives.

In this work, we prove that the self-correcting memory of the Wegner gauge theory is in fact stable to all small-enough perturbations, and thus constitutes a self-correcting memory phase. The method leverages a local symmetry of the unperturbed model. We show that even though dangerous perturbations $V$ like the surface tension break this symmetry, as long as their strength is small relative to the temperature, thermal fluctuations restore this symmetry exactly, allowing us to replace $V$ with an effective ``gauge averaged'' interaction $V_{\rm eff}$, which does not destabilize the memory. We summarize this as follows (see also Theorem \ref{thmWegnermemoryapp}):

\begin{figure}
    \centering
    \includegraphics[width=0.7\linewidth]{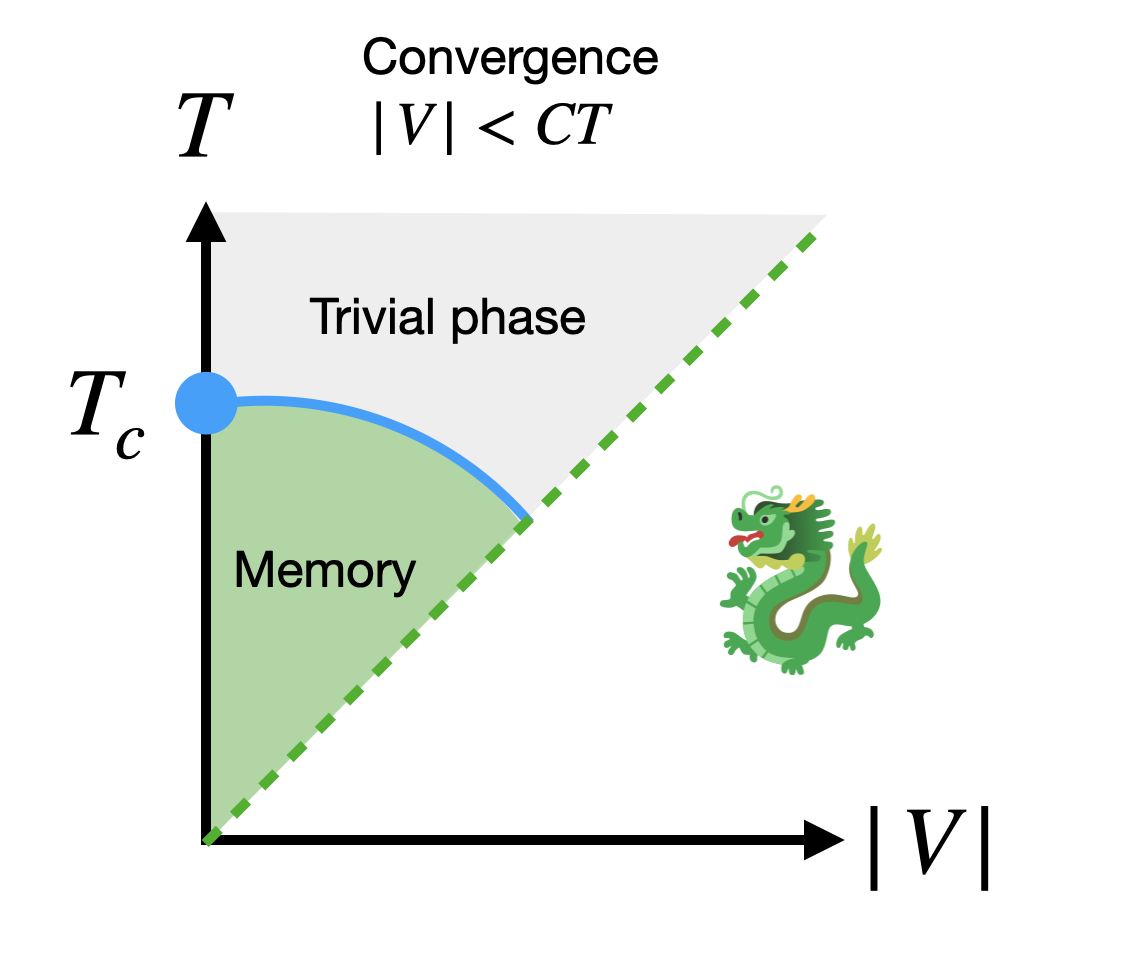}
    \caption{Schematic phase diagram of stable memory as a function of temperature $T$ and strength of perturbation $|V|$. $V_{\rm eff}$ defined by gauge averaging is local in the cone $|V| < CT$ (gray) and enjoys the local symmetry exactly. Here $C$ is a geometric constant. If there is a local-symmetry-protected memory phase at $V = 0$ up to some $T_c$, then the memory persists to a fluctuation-stabilized phase in a wedge, as shown schematically in green. Outside this region, denoted by the dragon, we can no longer use our method to define a local $V_{\rm eff}$, and the memory can break down.}
    \label{fig:wedge}
\end{figure}

\begin{thm}\label{thmWegnermemory}
    Let $H_0$ be the Hamiltonian\footnote{We represent our classical Gibbs states as quantum ones, where all Hamiltonians and perturbations are assumed to be diagonal in the classical basis.} \eqref{eqnWegnerham} of Wegner gauge theory on a 3-torus. For all temperatures $T = 1/\beta$ satisfying $0 < T < 1/\log 5$ and all parameters $\eta > \log 6 + 1 \approx 2.8$ there is a constant $C_{\rm stab}(T,\eta) >0$ such that for all perturbations $V$, if the interaction norm $|V|_\eta$ of $V$ defined in \eqref{eqninteractionnorm} satisfies
    \[\beta |V|_\eta < C_{\rm stab}(T,\eta),\]
    then the perturbed Gibbs state $\frac{1}{Z}e^{-\beta(H_0 + V)}$ has exponentially long-lived self-correcting memory.
\end{thm}

Thus, we find that the self-correction of Wegner gauge theory is in fact stabilized by fluctuations: the higher the temperature, the more robust it is to perturbations, as long as we are below some $T_c$. This implies a phase diagram as in Fig. \ref{fig:wedge}. This wedge qualitatively matches the numerically calculated phase diagram when $V$ is a magnetic field $V = -h \sum_e Z_e$ (equivalent to the surface tension) \cite{Poulin_2019,somoza2021self}. We can also use gauge averaging to make a quantitative comparison near the origin $T=h=0$. Our analysis there implies an upper bound for the slope of the critical line $S_c \le 1/\text{arctanh}(1/5) \approx 4.9$, while the numerically estimated value from \cite{somoza2021self} is $S_c \approx 4.5$.

Before getting into technical details, let us stress that although Wegner gauge theory is called a gauge theory, we will not consider its local symmetry as a redundancy in the configuration space \cite{kogut1979introduction}. It is just a model which happens to have extensively many symmetries. Gauge averaging says that this extensive symmetry is surprisingly robust: even when it is explicitly broken, it is restored exactly by fluctuations if the symmetry breaking is weak enough.

\textbf{Gauge Averaging---}Now we will introduce our main method, which replaces a symmetry breaking perturbation $V$ by a symmetry-preserving $V_{\rm eff}$.

Let $H_0$ be a Hamiltonian enjoying a finite symmetry group $G$, and let $V$ be an arbitrary perturbation. We consider the Gibbs state of the perturbed system $\rho = \frac{1}{Z} e^{-\beta (H_0+V)}$ where $Z = \Tr e^{-\beta(H_0+V)}$. We average $\rho$ over $G$ to define the gauge averaged state \footnote{Throughout, we work in finite volume and assume a finite number of local symmetries, so this is a finite sum. The method could be generalized to more general compact Lie group local symmetries.}:
\[\rho^G := \mathbb{E}_g(g\rho g^{-1}) = \frac{1}{|G|} \sum_{g \in G} g \rho g^{-1}.\]
This new state has all the same correlation functions for symmetric operators, since for any such operator $O$,
\[\Tr \rho^G O = \mathbb{E}_g \Tr g \rho g^{-1} O  =  \mathbb{E}_g \Tr  \rho O 
= \Tr \rho O\]
On the other hand, we can formally write
\[\label{eqngaugeaveragedrho}\rho^G = \frac{1}{Z} e^{-\beta (H_0 + V_{\rm eff})},\]
(note $Z$ is unchanged from above since $\Tr\ g e^{-\beta(H_0+V)} g^{-1} = \Tr\ e^{-\beta(H_0+V)}$). In the classical case where $H_0$ and $V$ commute, this effective perturbation is given by
\[\label{eqnclassicalgaugeaveraging}\mathbb{E}_g(g e^{-\beta V} g^{-1}) = e^{-\beta V_{\rm eff}}.\]
The main question is to determine whether there is a local expression for $V_{\rm eff}$. We will give such a local expression using a cluster expansion, which holds under a convergence condition on $V$.

Intuitively, since $H_0$ is $G$-symmetric, at any $T > 0$ the Gibbs state $Z_0^{-1} e^{-\beta H_0}$ assigns an equal weight to all configurations in the same $G$-orbit. Thus, it behaves as an infinite temperature distribution in the ``pure gauge'' directions of the configuration space. If we include a $G$-breaking perturbation $V$ whose energy scale is small compared to $T$, then these pure gauge directions will still be in a high temperature phase, and can be safely integrated out to obtain a $G$-symmetric $V_{\rm eff}$.

To state the Gauge Averaging Theorem, we need some preliminaries. We assume that $G$ is a finite abelian group and we choose a set $G_0 \subset G$ of ``local generators''. We also let $\Lambda$ be the set of local degrees of freedom. These local degrees of freedom can be any Hilbert space (recall we express classical models as quantum ones which are diagonal in the classical basis). These define a bipartite graph $\mathcal{G}$ whose vertices are $G_0 \cup \Lambda$, with $g_0 \in G_0$ and $x \in \Lambda$ sharing an edge iff the local generator $g_0$ acts non-trivially on the spin $x$. For $X \subset \Lambda$, define $G_0(X)$ as the set of local generators whose support overlaps $X$.

An interaction may be expressed as $V = \sum_{X \subset \Lambda} V_X$, which is a sum over operators $V_X$ supported in subsets $X$ of the set of spins $\Lambda$. For each parameter $\eta \in [0,\infty)$ we define the ``interaction norm''
\[\label{eqninteractionnorm}|V|_\eta := \sup_{g_0 \in G_0} \sum_{X \subset \Lambda: g_0 \in G_0(X)} ||V_X||_\infty e^{\eta l(X)}\]
where $||V_X||_\infty$ is the maximum value of $|V_X|$ over all spin configurations in $X$ and $l(X)$ is a measure of the diameter of $X$ in the graph $\mathcal{G}$ (see \eqref{eqndiameter}). For $\eta > 0$ this norm receives an enhanced contribution from larger size terms, forcing them to be exponentially small for the series to converge \footnote{Note that this norm depends on the decomposition of $V$ into local terms.}. This suppression factor is familiar from the study of correlation bounds in gapped systems \cite{Hastings_2006}.

\begin{thm}\label{thmclassicalaveraging}
    Let the local symmetry have the property that each local generator has support overlapping the supports of at most $\Delta$ other local generators. Then for all parameters $\eta \ge 0$, $m \ge 0$,$b > 0$ satisfying $ b < \eta-m-\log \Delta - 1$, there is a constant $C_{\rm gauge}(\eta,m,b) > 0$ such that if
    \[\label{eqnconvergencebasic} \beta |V|_\eta < C_{\rm gauge}(\eta,m,b),\]
    then the there is an exponentially-local symmetric effective interaction $V_{\rm eff}$ satisfying \eqref{eqngaugeaveragedrho} with
    \[\beta |V_{\rm eff}|_m < b.\]
\end{thm}
The proof is given in Theorem \ref{apptheoremstatement} in the appendix. As mentioned above, the local expressions $V_{{\rm eff},X}$ come from the cluster expansion and each term is $G$-symmetric. Note that the robustness to breaking the local symmetry actually increases with raising the temperature.

\textbf{Self-Correcting Memory in Wegner Gauge Theory (Proof Sketch of Theorem \ref{thmWegnermemory})---}The Wegner gauge theory \cite{Wegner1971duality} is a model of spin-1/2s living on the edges of a 3d cubic lattice, together with the Hamiltonian
\[\label{eqnWegnerham}H_0 = \sum_p \frac12(1-B_p) \qquad B_p = \prod_{e \in \partial p} Z_e\]
where $p$ is a plaquette, $e \in \partial p$ is an edge in the boundary of $p$, and $Z_e$ is the Pauli matrix $\sigma^z$ acting on the state of the spin at $e$. This system has the remarkable property, emphasized already by Wegner, of having a low temperature phase without a local order parameter. This low temperature phase may be diagnosed by the perimeter law of the expectation of the Wilson loop $W_\Gamma = \prod_{e \in \Gamma} Z_e$, where $\Gamma$ is a closed loop of edges in the cubic lattice.

More recently, it has been emphasized that the low temperature phase has topological sectors and displays self-correcting memory, meaning that under any local Markov dynamics satisfying detailed balance with respect to the Gibbs state (such as Glauber or Metropolis dynamics), distinct topological sectors take an exponentially-in-system-size long time to mix with one another \cite{Roberts_2017,Poulin_2019,stahl2026slow}. This property has been proved to be robust to adding a small field $V = - h\sum_e Z_e$ to the Hamiltonian in \cite{stahl2026slow} \footnote{Note the phase diagram as a function of temperature and magnetic field $h$ is dual to the quantum Fradkin-Shenker model (2+1d toric code in an in-plane field) \cite{fradkin1979phase}, a duality which was exploited to study the later model in \cite{tupitsyn2010topological,somoza2021self}.}. Here we will show it is robust to all small enough local perturbations, proving Theorem \ref{thmWegnermemory}. Full details will be given in the appendix as Theorem \ref{thmWegnermemoryapp}. Here we will sketch the general idea.

To prove there is an exponential memory lifetime, following \cite{stahl2026slow} we use the notion of a ``bottleneck''. This may be defined for a general distribution $\mu$ on a configuration space $\Omega$ and applies to any suitable class of local Markov dynamics satisfying detailed balance with respect to $\mu$. We say $\mu$ has a bottleneck if there is a subset $A \subset \Omega$ with $\mu(A) < 1/2$ and another subset $\partial A$, disjoint from $A$, such that any local Markov dynamics must pass through $\partial A$ to leave $A$. In this case, we get a lower bound on the mixing time for any local Markov dynamics satisfying detailed balance for $\mu$ given by
\[\label{eqnbottleneck}t_{\rm mix} \ge \frac14 \mu(A)/\mu(\partial A).\]
See Lemma \ref{lemmabottleneck} in the Appendix. Thus, if the ``bottleneck'' $\partial A$ is much smaller than the ``bottle'' $A$, the system gets stuck for a long time in $A$.

The topological sectors of Wegner gauge theory are defined as follows \cite{Castelnovo_2008,serna2024worldsheetpatching1formsymmetries,stahl2026slow}. We consider a 3d cubic lattice with periodic boundary conditions of size $L \times L \times L$. We turn our spin configuration $Z_e = \pm 1$ into a $\bZ_2$ 2-chain $S$ on the dual cubic lattice, taking $S$ to be the mod 2 sum of dual plaquettes $p$ whose associated edge $e$ has $Z_e = -1$. The unperturbed energy \eqref{eqnWegnerham} is the length of the boundary of this 2-chain $\partial S$, which is a 1-cycle on the dual cubic lattice we call the syndrome. We say that two 1-cycles on the dual cubic lattice are disjoint if they do not share any vertices. We can express $\partial S$ as a disjoint union of 1-cycles. If all of these 1-cycles are short enough, say they each have length $< L$, then we ``decode'' the topological sector of $S$ as follows. Each connected component $\gamma_i \subset \partial S$ of the syndrome may be contained in a ``bounding cube'' of side length $L/2$, and within this cube we may choose a surface $D_i$ whose boundary $\partial D_i =\gamma_i$. Taking $C=S + \sum_i D_i$ we obtain a 2-cycle, which has a definite homology class independent of our choices of $D_i$, since each bounding cube is contractible. For simplicity, choose a particular bounding surface $D_i = D(\gamma_i)$ for all such loops $\gamma_i$ once and for all.

Let us suppose for the sake of the current discussion that the updates of our Markov chain act only on a single edge at a time (we relax this in Theorem \ref{thmWegnermemoryapp} below). Then if the configuration $S$ has all components of $\partial S$ length $< L/4$, then after a single Markov update, all of its components have a length $<L$. Indeed, the most these components can increase in length is when four components of $\partial S$ merge. Furthermore, such a single-edge update can create at most one loop of length $\ge L/4$.

We define the bottle $A[C]$ to be the set of configurations $S$ such that the components of $\partial S$ all have length $<L/4$ (call these components ``short'') and the ``decoded'' homology class of $S$ is $[C] \in H_2(T^3,\bZ_2) = \bZ_2^3$. We also define the bottleneck $\partial A[C]$ as the set of configurations $S$ for which precisely one component of the syndrome $\partial S$ is a ``long loop'' of length in $[L/4,L)$ and all others have length $< L/4$, and such that if we decode this configuration we obtain the class $[C]$. By definition, the Markov chain must pass through $\partial A[C]$ to leave $A[C]$.

Now that we have the bottles and the bottlenecks, we want to lowerbound the bottleneck ratio \eqref{eqnbottleneck}. We do this for the unperturbed model to demonstrate the method, following \cite{stahl2026slow}. Let $\Gamma$ be the long loop of $S \in \partial A[S]$. Recall we have chosen a particular $D(\Gamma)$ in its bounding cube such that $\partial D(\Gamma) = \Gamma$. Since all other components of $\partial S$ are short, $S+D(\Gamma) \in A[C]$. Thus, we obtain an injective map
\[i:\partial A[C] \hookrightarrow A[C] \times \{\text{loops of length }l\in[L/4,L)\}.\]
Moreover, the energy of $S$ is $|\partial S| = |\partial S+D(\Gamma)| + |\partial \Gamma|$. Thus, the Gibbs weight of $\mu(\partial A[C])$ can be bounded above by a sum over pairs on the right-hand side above.

The long loops may be overcounted by considering paths of length $l \in [L/4,L)$, giving us a Peierls bound \cite{peierls1936Ising} (assuming low enough temperature $5e^{-\beta} < 1$)
\[\mu(\partial A[C]) \le \mu(A[C]) L^3 \sum_{l=\lceil L/4 \rceil}^{L-1} 5^l e^{-\beta l} \\ \le \mu(A[C]) L^3 \frac{(5 e^{-\beta})^{L/4}}{1-5e^{-\beta}}\]
Since there are 8 possible $[C]$ on $T^3$, and each $A[C]$ is disjoint, at least one of them has $\mu(A[C]) < 1/2$. Therefore we can apply the bottleneck lemma, and we find a mixing time which diverges exponentially as $(e^{\beta}/5)^{L/4}/L^3$. Note that this gives a lower bound on the temperature above which the model is in a trivial phase: $T_c \ge 1/\log 5 \approx 0.62$. The numerically measured critical point lies at $T_c \approx 0.65$ \cite{Agrawal_2025}, suggesting that it is the entropic deconfinement of these long loops which indeed drives the phase transition. See also \cite{Roberts_2017} who derived the same estimate.

For the perturbed model, this Peierls analysis is in danger of breaking down, since $V$ may depend on more than just $\partial S$. In particular, the magnetic field $-h \sum_e Z_e$ gives a ``surface tension'' contribution proportional to the weight of $S$ itself, which invalidates the loop counting above. Indeed, at low temperature, one expects that the Gibbs distribution is dominated by low energy states, and the lowest energy states in non-trivial topological sectors $[C] \neq 0 \in H_2(T^3,\bZ_2)$ have a weight $|S| \approx L^2$. From this point of view, perturbative stability of the memory seems very surprising. 

However, gauge averaging provides the means to understand how perturbative stability arises from fluctuations, with a region of stability which grows with the temperature as in \eqref{eqnconvergencebasic}. Indeed, the unperturbed model has a local symmetry with generators associated to each vertex $v$, given by the operator
\[A_v = \prod_{e: v \in \partial e} X_e\]
where the product is over edges $e$ which have $v$ as a boundary vertex and $X_e$ is the Pauli matrix $\sigma^x$ acting on the spin at $e$. The associated bipartite graph $\mathcal{G}$ is the graph of vertices and edges of the cubic lattice. This has a bounded degree $\Delta = 6$ and so Theorem \ref{thmclassicalaveraging} applies. In particular, for any weak enough local perturbation $V$, there will be an exponentially-local effective perturbation $V_{\rm eff}$.

Let $\mu$ be the perturbed Gibbs distribution and $\mu^G$ its gauge averaged version. The sets $A[C]$ and $\partial A[C]$ are symmetric sets under $G$, defined only in terms of $\partial S$ and its topological sector (its homology class relative to $\partial S$), and thus
\[\label{eqnbottleneckequality}\mu(A[C])/\mu(\partial A[C]) = \mu^G(A[C])/\mu^G(\partial A[C]).\]
Therefore, if we can prove a bottleneck in these sets for $\mu^G$, we can conclude the exponential memory lifetime for $\mu$.

We have to make sure the Peierls argument applies to $\mu^G$. The important quantity to control is
\[|V_{\rm eff}(S' + D(\Gamma)) - V_{\rm eff}(S')|\]
where $S' \in A[C]$ and $\Gamma$ is a loop of length $[L/4,L)$. Because $V_{\rm eff}$ is symmetric, contributions to this can only come from terms whose support contains a curve linking $\Gamma$. Using the norm bound for $V_{\rm eff}$ in Theorem \ref{thmclassicalaveraging}, we are able to show that there is a perimeter law bound $\le C_{\rm perim}(m) |V_{\rm eff}|_m |\Gamma|$+exp. small, where $C_{\rm perim}(m)$ is a system-size-independent constant. Thus, for small enough $|V|_\eta$, we indeed get the required bottleneck, therefore proving the stability of the self-correcting memory. See Theorem \ref{thmWegnermemoryapp} in the Appendix for more details.

\textbf{Special Case: Magnetic Field Perturbation---}In the special case of an applied field $V = -h \sum_e Z_e$, we can give a relatively explicit form of $V_{\rm eff}$, and obtain a better estimate on the stability of the memory phase for this particular perturbation than Theorem \ref{thmclassicalaveraging}, which applies to an arbitrary perturbation.

First, since $Z_e^2=1$, we have $e^{\beta h Z_e} = \cosh(\beta h) + \sinh(\beta h) Z_e$. Up to a constant, we can thus express
\[\label{eqnfieldsum}e^{\beta h \sum_e Z_e} \propto \prod_e (1 + \tanh (\beta h) Z_e) = \sum_{\Gamma \subset E} \tau^{|\gamma|} \prod_{e \in \Gamma} Z_e\]
where $E$ is the set of edges of the cubic lattice and $\tau = \tanh(\beta h)$ is the natural small parameter in the sum expression. Gauge-averaging this expression projects this expression to a sum over edge sets $\Gamma$ which meet each vertex an even number of times.

We say such even edge sets are independent if they meet no common vertices. A ``polymer'' will be an even edge set which cannot be divided into two independent even edge sets. Let $P$ be the set of polymers and for $\gamma \in P$ define $W(\gamma) = \tau^{|\gamma|} \prod_{e \in \gamma} Z_e$. The cluster expansion yields (formally)
\[\label{eqnmagfieldcluster}V_{\rm eff} = -\frac{1}{\beta}\sum_{n=1}^\infty \frac{1}{n!} \prod_{(\gamma_1,\ldots,\gamma_n) \in P^n} \phi^T(\gamma_1,\ldots,\gamma_n) \prod_{i=1}^n W(\gamma_i)\]
where $\phi^T(\gamma_1,\ldots,\gamma_n)$ is a combinatorial factor which vanishes unless $(\gamma_1,\ldots,\gamma_n)$ is a ``cluster'' meaning it has a connected dependency graph, see \cite{fernandez2007cluster}.

We can apply the proof of Theorem \ref{thmclassicalaveraging} to this sum to obtain a lower bound on its radius of convergence. However, we can also observe that \eqref{eqnfieldsum} has the same expression as the thermal partition function for the unperturbed Wegner gauge theory we discussed above, but with $\gamma$ playing the role of the syndrome $\partial S$ and $\tau = e^{-\beta}$. This is an aspect of electric-magnetic duality in this model. Exploiting this, the Peierls analysis above implies that the generated Wilson loops in $V_{\rm eff}$ remain confined as long as $\tau < 1/5$. Even so, as analyzed above, they shift the effective tension of the syndrome $\partial S$. The size of the shift is $O(1/\beta)$ at each fixed $\tau$, so if we consider the effect of the perturbation near the origin, taking the limit $h \to 0$, $\beta \to \infty$ but keeping the ratio $\beta h$, hence $\tau$, fixed, then as long as the Wilson loops are confined, they shift the tension of the syndrome by a vanishing amount, and so the phase transition line at the origin must satisfy $\tau_c \ge 1/5$. Meanwhile, the numerical estimate extracted from \cite{somoza2021self} is $\tau_c \approx 0.22$. Thus, gauge averaging puts us within 10\% of the estimated value.

\textbf{Discussion---}In this work, we have introduced gauge averaging as a method for proving the perturbative stability of self-correcting classical memory phases. We find these phases display a mixture of low temperature physics along the gauge-invariant directions of configuration space (convergent Peierls bounds) and high temperature physics along the pure-gauge directions (convergent cluster expansion). They are fluctuation-stabilized phases of matter, in the sense that their perturbative stability increases with $T$ for small $T$. It would be interesting to know if this property is at all related to the fluctuation-stabilized non-equilibrium memories studied in \cite{lake2026squeezingcodesrobustfluctuationstabilized}.

Gauge averaging demonstrates that extensive local symmetries are surprisingly robust in classical thermal ensembles. This can be thought of as a cousin of Elitzur's theorem \cite{elitzur} which says that these symmetries cannot be spontaneously broken. Our theorem says they are even resistant to explicit breaking.

In the case of Wegner gauge theory, the product of many local symmetry generators telescopes to form a boundary symmetry along a surface, which can be interpreted as a 1-form symmetry of the model \cite{Gaiotto_2015}. Thus, gauge averaging gives a precise sense in which the 1-form symmetry re-emerges exactly in $V_{\rm eff}$. It would be very interesting to understand the relationship between this phenomenon and recent proposals for defining emergent 1-form symmetries \cite{Pace_2023,serna2024worldsheetpatching1formsymmetries,liu2026informationtheoreticprincipleemergent1form}.

\textit{Acknowledgements---}I am grateful to Shankar Balasubramanian, Margarita Davydova, Thomas Dumitrescu, Lei Gioia, and Ting-Chun (David) Lin, and Ruben Verresen for providing useful feedback and related collaborations. I am also grateful to acknowledge the support of the Bhaumik Presidential Term Chair.

\bibliography{main}

\appendix
\counterwithout{equation}{section}

\onecolumngrid

\section{Classical Gauge Averaging and Cluster Expansions}

\subsection{Definitions}\label{subsecdefs}

\begin{defin}\label{deflocalsymmetry}
    A local symmetry acting on a set $\Lambda$ of degrees of freedom is a pair $(G,G_0)$ of finite abelian group $G$ and a generating set $G_0$
    \begin{itemize}
        \item For each $x \in \Lambda$ there is some local symmetry $g \in G$ whose support overlaps $x$.
        \item For each operator $V_X$ which is diagonal in the classical basis and supported in $X$ and each $g \in G$, $g V_X g^{-1}$ is also diagonal in the classical basis and supported in $X$.
        \item The symmetry graph $\mathcal{G}$ is connected.
    \end{itemize}
\end{defin}
The first condition is for simplicity of notation. If there are other degrees of freedom that the local symmetry does not act on, we can simply ignore them in the average. The second is for the simple expression of operator supports. Both of these conditions hold for the symmetry of Wegner gauge theory.

We start with an interaction, expressed as a sum over local interactions $V_X$ supported in subsets $X \subset \Lambda$:
\[\label{eqnlocalinteraction}V = \sum_{X \subset \Lambda} V_X\]
By convention we will take $V_\varnothing = 0$, since this constant does not affect thermodynamic quantities. We want to compute the gauge-averaged effective interaction
\[V_{\rm eff} = -\frac{1}{\beta} \log \mathbb{E}_g (ge^{-\beta V}g^{-1})\]
First we write
\[e^{-\beta V} = \prod_{X \subset \Lambda} e^{-\beta V_X} = \prod_{X \subset \Lambda} (1+f(X,\beta)) = \sum_{n=0}^\infty\Sum_{\{X_1,\ldots,X_n\} \subset 2^\Lambda} \prod_{j=1}^n f(X_j,\beta)\]
where $f(X,\beta) = e^{-\beta V_X}-1$ and $2^\Lambda$ is the power set of $\Lambda$. Linearity of expectation gives
\[e^{-\beta V_{\rm eff}} = \sum_{n=0}^\infty\Sum_{\{X_1,\ldots,X_n\} \subset 2^\Lambda} \mathbb{E}_g\left(g \left(\prod_{j=1}^n f(X_j,\beta)\right) g^{-1}\right).\]

For $X \subset \Lambda$, let $G_0(X)$ be the set of local symmetry generators whose support overlaps $X$. We say two $X,Y \subset \Lambda$ are  ``dependent'' if they have a common overlapping symmetry generator:
\[X \sim Y \Leftrightarrow G_0(X) \cap G_0(Y) \neq \varnothing\]
Since $G$ is finite abelian and generated by $G_0$, the uniform measure on $G$ is the pushforward of the product measure on the abstract product of cyclic groups generated by the elements of $G_0$. Thus functions depending on disjoint subsets of $G_0$ will be independent under gauge averaging. Therefore the expectation value above splits into a product over sets of subsets $\gamma=\{X_1,\ldots,X_m\}$ with connected dependency graphs, which is the graph whose vertices are $\{X_i\}_{i=1,\ldots,m}$ and have edges whenever $X_i \sim X_j$.

We call a set of subsets $\gamma \subset 2^\Lambda$ with a connected dependency graph a ``polymer''. Let $P$ be the set of polymers. We say two polymers $\gamma,\gamma' \in P$ are dependent if
\[\gamma \sim \gamma' \Leftrightarrow \exists X \in \gamma, \exists Y \in \gamma', X \sim Y.\]
(Note that many references, including \cite{kotecky1986cluster}, call this relation ``incompatibility'' and denote it with $\not\sim$.) Otherwise we say they are independent. We can then write
\[e^{-\beta V_{\rm eff}} = \sum_{n=0}^\infty \frac{1}{n!}\Sum_{(\gamma_1,\ldots,\gamma_n) \in P^n \\ \text{mutually independent polymers}} \prod_{j=1}^n \Phi(\gamma_j) \\
\Phi(\gamma) = \mathbb{E}_g \left(g \left(\prod_{X \in \gamma} f(X,\beta)\right) g^{-1}\right)\]
Evaluated in the classical basis, where $H_0$ and $V$ are diagonal, for each fixed configuration of the spins, we can consider this as an abstract polymer model in the language of Koteck\'y and Preiss \cite{kotecky1986cluster}.

Under certain conditions stated in \cite{kotecky1986cluster,fernandez2007cluster}, $V_{\rm eff}$ evaluated on a particular spin configuration has an absolutely convergent cluster expansion. We will apply this cluster expansion pointwise over all spin configurations to get a cluster expansion for $V_{\rm eff}$, while convergence results will be uniform in the spin configurations themselves.

A cluster is an ordered tuple of polymers $C=(\gamma_1,\ldots,\gamma_n) \in P^n$ (may include repeats) with a connected dependency graph, a graph whose vertices are the polymers $\gamma_1,\ldots,\gamma_n$, with an edge whenever $\gamma_i \sim \gamma_j$ are dependent. The cluster expansion is, formally
\[\label{eqnclusterexpapp}-\beta V_{\rm eff} = \sum_{n=1}^\infty \frac{1}{n!} \Sum_{C \in P^n \\ \text{cluster} } \Phi^T(C) \\
\Phi^T(C)=\phi^T(C) \prod_{\gamma \in C} \Phi(\gamma)\]
where $\phi^T(C)=1$ if $n=1$, and otherwise
\[\phi^T(C) = \Sum_{F \subset E(C) \\ F\text{ connected, spanning}} (-1)^{|F|}\]
where $E(C)$ is the set of edges in the dependency graph of $C$, and $F$ is a subset of these edges such that every vertex $\gamma \in C$ has at least one edge in $F$ connected to it and the resulting subgraph is connected. See \cite{fernandez2007cluster} for this explicit formula and a nice review of various convergence criteria.

To place $V_{\rm eff}$ into the general form \eqref{eqnlocalinteraction} of an interaction, we define for a cluster $C = (\gamma_1,\ldots,\gamma_n)$ define the ``support'' of a cluster $C$ as
\[S(C) = \bigcup_{j=1}^n \bigcup_{X \in \gamma_j} X\]
and define
\[\label{eqnlocalinteractioneffective}V_{{\rm eff},X} = - \frac{1}{\beta} \sum_{n \ge 1} \frac{1}{n!} \Sum_{C \in P^n \\ \text{cluster, }S(C)=X} \Phi^T(C).\]
Note that because of the second bullet of Definition \ref{deflocalsymmetry}, this operator has support contained in $X$.

To analyze the convergence of the cluster expansion for gauge averaging, we need several more definitions. As in the main text, we consider the bipartite graph $\mathcal{G}$ with vertex and edge sets
\[V(\mathcal{G}) = \Lambda \cup G_0 \\
E(\mathcal{G})= \{(x,g_0)\ |\ \text{$x$ is in the support of $g_0$}\}\]
For $X \subset \Lambda$ we define
\[G_0(X) := \{\text{set of $g_0 \in G_0$ connected to $X$}\}\]
and for a polymer $\gamma \in P$ we define
\[G_0(\gamma) := G_0(\bigcup_{X \in \gamma} X)\]
We also define two measures of the size of a polymer. The first is the number of gauge generators with support overlapping that of $\gamma$:
\[||\gamma|| := |G_0(\gamma)|= |\bigcup_{X \in \gamma} G_0(X)|\]
To define the second, for each $X \subset \Lambda$ we choose a ``connected hull'', which is a set $H(X) \subset \Lambda$ such that $X \subset H(X)$, $H(X) \cup G_0(H(X))$ is connected in $\mathcal{G}$, and $|G_0(H(X))|$ is minimized over all such sets. Also define
\[\label{eqndiameter}H(\gamma) := \bigcup_{X \in \gamma} H(X) \\
l(X) := |G_0(H(X))| \\
l(\gamma) := |G_0(H(\gamma))|\]
For $\eta \in [0,\infty)$ we define the interaction norm
\[|V|_\eta = \sup_{g_0 \in G_0} \Sum_{X \subset \Lambda \\ g_0 \in G_0(X)} ||V_X||_\infty e^{\eta l(X)},\]
where $||V_X||_\infty$ is the maximum value of $|V_X|$ over all configurations of spins in $X$. This norm is very similar to the one used in \cite{Hastings_2006}, with $l(X)$ serving as an abstract measure of the ``diameter'' of $X$. Intuitively, the $e^{\eta l(X)}$ factor captures the effect that larger weight perturbations are able to generate huge weight operators at a lower order in perturbation theory than smaller weight perturbations, giving the former an exponential boost in destructive power. $\eta^{-1}$ determines a length scale over which the operators in $V$ are required to decay so that $|V|_\eta$ remains finite in the thermodynamic limit.

\subsection{Proof of the Gauge Averaging Theorem}

\begin{thm}\label{apptheoremstatement}
    Let $(G,G_0)$ be a local symmetry satisfying Definition \ref{deflocalsymmetry}. Suppose it has the property that each local generator $g_0 \in G_0$ has support overlapping the supports of at most $\Delta$ other local generators. Then for all $0 < b < \eta-m-\log \Delta - 1$, there is a constant $C_{\rm gauge}(\eta,m,b) > 0$ such that if
    \[\beta |V|_\eta < C_{\rm gauge}(\eta,m,b),\]
    then the cluster expansion is absolutely convergent and gives an exponentially-local interaction \eqref{eqnlocalinteractioneffective} with
    \[\beta |V_{\rm eff}|_m < b.\]
\end{thm}

\begin{proof}
In order to prove convergence, we will consider the abstract polymer model
\[\sum_{n = 0}^\infty \frac{1}{n!} \Sum_{(\gamma_1,\ldots,\gamma_n) \in P^n \\ \text{mutually independent polymers}} \prod_{j=1}^n z(\gamma_j) \\
z(\gamma_j) = ||\Phi(\gamma_j)||_\infty.\]
The absolute cluster expansion of this model dominates the cluster expansion \eqref{eqnclusterexpapp} for every spin configuration. Therefore, absolute convergence for this cluster expansion will prove uniform absolute convergence for \eqref{eqnclusterexpapp}.

Define for each polymer $\gamma \in P$ the ``KP activity''
\[K_{\gamma,b} = \sum_{\gamma' \in P: \gamma' \sim \gamma} ||\Phi(\gamma')||_\infty e^{ml(\gamma') + b||\gamma'||}\]
The ``KP criterion'' is
\[\label{eqnKPcrit}\forall \gamma, K_{\gamma,b} \le b||\gamma|| \qquad \text{KP criterion}\]
Under this condition, absolute convergence of the cluster expansion is guaranteed by the main theorem of \cite{kotecky1986cluster} with $a(\gamma) = b||\gamma||$ and $d(\gamma) = m l(\gamma)$. In the notation of \cite{fernandez2007cluster}, which will be useful below, we have, for $\gamma \in P$
\[\rho_\gamma = ||\Phi(\gamma)||_\infty e^{ml(\gamma)} \\
a_\gamma = b ||\gamma||\]
With this notation \eqref{eqnKPcrit} is Eq. 2.15 of \cite{fernandez2007cluster}.

We want to derive the KP criterion \eqref{eqnKPcrit} from our assumptions. It will be easier to upper bound the ``pinned'' KP activity
\[K_{b,m} = \sup_{g_0 \in G_0} \Sum_{\gamma \in P \\ g_0 \in G_0(\gamma)} ||\Phi(\gamma)||_\infty e^{b||\gamma|| + ml(\gamma)}.\]
First we will show that if $K_{b,m} \le b$, then the KP criterion \eqref{eqnKPcrit} is satisfied. Indeed, if $\gamma' \sim \gamma$, then there is some $g_0 \in G_0(\gamma) \cap G_0(\gamma')$. Thus,
\[K_{\gamma,b} \le \sum_{g_0 \in G_0(\gamma)} \sum_{\gamma': g_0 \in G_0(\gamma')} e^{b||\gamma'||+ml(\gamma')} ||\Phi(\gamma')||_\infty \le ||\gamma|| K_{b,m}.\]
Thus, if we can show
\[\label{eqnpinnedcriterion}K_{b,m} \le b \quad \text{pinned KP criterion}\]
then the KP criterion \eqref{eqnKPcrit} will be satisfied.

Now to upper bound $K_{b,m}$. Let
\[\rho_X = ||f(X,\beta)||_\infty = ||e^{-\beta V_X}-1||_\infty\]
Since $||\Phi(\gamma)||_\infty \le \prod_{X \in \gamma} \rho_X$, we have
\[K_{b,m} \le \sup_{g_0 \in G_0} \Sum_{\gamma \in P \\ g_0 \in G_0(\gamma)} e^{b||\gamma|| + ml(\gamma)} \prod_{X \in \gamma} \rho_X.\]
Note that
\[l(\gamma) \le \sum_{X \in \gamma} l(X) \\
||\gamma|| \le l(\gamma)\]
so for all $\eta \ge 0$,
\[K_{b,m} \le \sup_{g_0 \in G_0} \Sum_{\gamma \in P \\ g_0 \in G_0(\gamma)} e^{-(\eta-b-m) l(\gamma)} \left(\prod_{X \in \gamma} \rho_X e^{\eta l(X)}\right)\]
We now want to group the sum over polymers with a fixed $l(\gamma)$.

Recall this is defined using connected hulls $H(\gamma) \subset \Lambda$ in \eqref{eqndiameter}. For this next bound, we want to use ``connected'' sets of local symmetry generators. Define, for $J \subset G_0$ the set $\Lambda(J)$ of spins in the support of elements in $J$. We will say $J$ is ``connected'' if $J \cup \Lambda(J)$ is connected in $\mathcal{G}$. We define, for each $X \subset \Lambda$ and polymer $\gamma$,
\[J(X) = G_0(H(X)) \\
J(\gamma) = \bigcup_{X \in \gamma} J(X)\]
Note that $J(\gamma)$ is connected since each $J(X)$ is, and since $X \sim Y$ implies $J(X) \cap J(Y) \neq \varnothing$. We have by definition
\[l(\gamma) = |J(\gamma)|.\]
Let
\[\mathcal{J}(l)=\{J \subset G_0\ |\ J\text{ connected and }|J|=l\}.\]
We have
\[K_{b,m} \le \sum_{l=1}^\infty e^{-(\eta-b-m) l} \sup_{g_0 \in G_0}\Sum_{J \in \mathcal{J}(l) \\ g_0 \in J} \Sum_{\gamma \in P \\ J(\gamma)=J} \prod_{X \in \gamma} \rho_X e^{\eta l(X)}\]
Define
\[\mathcal{X}(J) = \{X \neq \varnothing\ |\ J(X) \subset J\}.\]
Polymers $\gamma$ with $J(\gamma) = J$ must be made of sets $X \in \mathcal{X}(J)$. We can thus replace the sum over $\gamma$ with a sum over (nonempty) subsets of $\mathcal{X}(J)$:
\[K_{b,m} \le \sum_{l=1}^\infty e^{-(\eta-b-m) l} \sup_{g_0 \in G_0}\Sum_{J \in \mathcal{J}(l) \\ g_0 \in J} \Sum_{\Gamma \in 2^{\mathcal{X}(J)} \\ \Gamma \neq \varnothing} \prod_{X \in \Gamma} \rho_X e^{\eta l(X)}\]
We can replace the sum of products with an exponential of a sum:
\[\Sum_{\Gamma \in 2^{\mathcal{X}(J)} \\ \Gamma \neq \varnothing} \prod_{X \in \Gamma} \rho_X e^{\eta l(X)} \le \exp\left( \Sum_{X \in \mathcal{X}(J)} \rho_X e^{\eta l(X)} \right) - 1\]
Since $G_0(X) \subset J(X) \subset J$,
\[\sum_{X \in \mathcal{X}(J)}\rho_X e^{\eta l(X)}  \le \sum_{g \in J} \sum_{X: g \in G_0(X)} \rho_X e^{\eta l(X)} \le l q_\eta\]
where we used $l = |J|$ and defined
\[q_\eta = \sup_{g_0 \in G_0} \sum_{X:g_0 \in G_0(X)} \rho_X e^{\eta l(X)}.\]
Note that since $\rho_X \le e^{\beta ||V_X||}-1$,
\[\label{eqnqestimate}q_\eta \le 2 \beta|V|_\eta \quad \text{whenever} \quad \beta |V|_\eta \le \log 2 \]
and in particular goes to zero as $\beta|V|_\eta \to 0$. It will serve as an intermediate measure of the interaction norm.

Putting it all together we have
\[K_{b,m} \le \sum_{l=1}^\infty A(l) e^{-(\eta-b-m)l} (e^{q_\eta l}-1) \\
A(l) = \sup_{g_0 \in G_0} |\{J \in \mathcal{J}(l) : g_0 \in J\}|\]
This bound is true without any assumptions on $\mathcal{G}$. If however $\mathcal{G}$ has bounded symmetry degree $\Delta$, then (see Lemma 9 of \cite{mcdiarmid2021componentstructuredenserandom})
\[A(l) \le (e \Delta)^l\]
Write
\[\alpha = e^{(\log \Delta + 1)-(\eta-b-m)}\]
For whenever
\[\eta-m > \log \Delta + 1\]
then for all $b \in (0,\eta-m-\log \Delta-1)$ we have $\alpha < 1$. In this case, there is some constant $C''(\eta,m,b)$ such that
\[q_\eta < C''(\eta,m,b) \implies \alpha e^{q_\eta} < 1\]
which then implies (summing the above inequality)
\[K_{b,m} \le \frac{\alpha e^{q_\eta}}{1- \alpha e^{q_\eta}} - \frac{\alpha}{1-\alpha}\]
This is a continuous function of $q_\eta$ and goes to zero as $q_\eta \to 0$, therefore there is a another constant $C'(\eta,m,b) \le C''(\eta,m,b)$ such that for all $b \in (0,\eta-m-\log \Delta -1)$,
\[q_\eta < C'(\eta,m,b) \implies K_{b,m} < b.\]
and thus by the argument above \eqref{eqnpinnedcriterion} the KP criterion will be satsified. We will take the constant in the theorem statement to be
\[C_{\rm gauge}(\eta,m,b) = \min(\log 2, C'(\eta,m,b)/2)\]
since by \eqref{eqnqestimate}
\[\beta |V|_\eta < C_{\rm gauge}(\eta,m,b) \implies q_\eta < C'(\eta,m,b)\]

We now wish to show the bound on the interaction norm of $V_{\rm eff}$. This essentially follows the analysis in \cite{fernandez2007cluster}.
We want to compute
\[|V_{\rm eff}|_m=\sup_{g_0 \in G_0} \Sum_{X \subset \Lambda \\ g_0 \in G_0(X)} ||V_{{\rm eff},X}||_\infty e^{m l(X)}\]
The local form of the cluster expansion is given in \eqref{eqnlocalinteractioneffective}. Note that
\[||\sum_{C:S(C)=X} \Phi^T(C)||_\infty \le \sum_{C:S(C)=X} ||\Phi^T(C)||_\infty \\
G_0(S(C)) = \bigcup_{j=1}^n G_0(\gamma_j) \\ l(S(C)) \le \sum_{j=1}^n l(\gamma_j),\]
where the last holds because $\bigcup_{j=1}^n H(\gamma_j) \cup G_0(H(\gamma_j))$ is connected and contains $S(C)$. Putting these together, we have
\[\beta |V_{\rm eff}|_m \le \sup_{g_0 \in G_0} \sum_{n = 1}^\infty \frac{1}{n!} \Sum_{(\gamma_1,\ldots,\gamma_n) \in P^n \\ \text{cluster}} ||\Phi^T(C)||_\infty e^{m \sum_{j=1}^n l(\gamma_j)} \sum_{j=1}^\infty \textbf{1}_{g_0 \in G_0(\gamma_j)} \]
where $\textbf{1}_{g_0 \in G_0(\gamma_j)}$ is a function which is 1 if $g_0 \in G_0(\gamma_j)$ and 0 otherwise. By permutation symmetry of the $\gamma_j$, we can absorb the sum over $j$ of the indicator functions:
\[\beta |V_{\rm eff}|_m \le \sup_{g_0 \in G_0} \sum_{n=1}^\infty \frac{1}{(n-1)!}\Sum_{(\gamma_1,\ldots,\gamma_n) \in P^n \\ \text{cluster, } g_0 \in G_0(\gamma_1)} ||\Phi^T(\gamma_1,\ldots,\gamma_n)||_\infty e^{m\sum_{j=1}^n(\gamma_j)} \]
Furthermore,
\[||\Phi^T(\gamma_1,\ldots,\gamma_n)||_\infty \le |\phi^T(\gamma_1,\ldots,\gamma_n)| \prod_{j=1}^n ||\Phi(\gamma_j)||_\infty\]
So separating off $\gamma_1$, we get
\[\beta |V_{\rm eff}|_m \le \sup_{g_0 \in G_0} \sum_{\gamma_1: g_0 \in G_0(\gamma_1)} ||\Phi(\gamma_1)||_\infty e^{ml(\gamma_1)} \left(1 + \sum_{n=1}^\infty \frac{1}{n!} \sum_{(\gamma_2,\ldots,\gamma_{n+1}) \in P^n} |\phi^T(\gamma_1,\ldots,\gamma_{n+1})| \prod_{j=2}^{n+1} ||\Phi(\gamma_j)||_\infty e^{ml(\gamma_j)} \right)\]
Recalling $\rho_\gamma = ||\Phi(\gamma)||_\infty e^{m l(\gamma)}$, the expression in parentheses is called $\Pi_{\gamma_1}$ in \cite{fernandez2007cluster} Eq. 2.7. They show under the KP criterion that
\[\rho_\gamma\Pi_\gamma \le \rho_\gamma e^{a_\gamma}\]
(see Eq. 2.14 with $\mu_\gamma = \rho_\gamma e^{a_\gamma}$, and also see below Eq. 2.15). This gives
\[\beta |V_{\rm eff}|_m \le \sup_{g_0 \in G_0} \sum_{\gamma: g_0 \in G_0(\gamma)} ||\Phi(\gamma)||_\infty e^{m l(\gamma) + b||\gamma||} = K_{b,m}\]
With our assumptions above we have $K_{b,m} < b$ proving the result.

\end{proof}

\section{Proof of Stability of Wegner Gauge Theory}

Let $\Omega$ be a set of configurations, $P$ the transition matrix of a Markov chain on $\Omega$, and $\mu$ a distribution on $\Omega$. We say $P$ is $\mu$-reversible if it satisfies detailed balance with respect to $\mu$, i.e.
\[\forall \sigma,\sigma' \in \Omega \quad \mu(\sigma) P(\sigma,\sigma') = \mu(\sigma') P(\sigma',\sigma).\]
We define the mixing time as
\[t_{\rm mix} = \inf\{t: \max_{\sigma \in \Omega} ||P^t(\sigma) - \mu||_{TV} \le 1/4\}, \qquad \inf(\varnothing)=\infty\]
where $P^t(\sigma)$ is $P$ applied $t$ times to the initial state $\sigma$ and $||\cdot||_{TV}$ is the total variation distance. The bottleneck lemma says the following:
\begin{lemma}\label{lemmabottleneck} \textbf{Bottleneck lemma:}
    Suppose $A \subset \Omega$ has $0<\mu(A) \le 1/2$ and that there is $\partial A \subset A^c$ such that if $\sigma' \notin A \cup \partial A$, then for all $\sigma \in A$,
    \[\label{eqnbottleneckcondition}P(\sigma,\sigma') = 0.\]
    Then if $P$ is $\mu$-reversible,
    \[t_{\rm mix} \ge \frac14 \frac{\mu(A)}{\mu(\partial A)}\]
\end{lemma}
\begin{proof}
    For $A,B \subset \Omega$ we can consider the conductance
    \[Q(A,B) = \sum_{\sigma \in A, \sigma' \in B} \mu(\sigma)P(\sigma,\sigma').\]
    If every exit from $A$ lands in $\partial A$, then
    \[Q(A,A^c) = Q(A,\partial A).\]
    By $\mu$-reversibility,
    \[Q(A,\partial A) = Q(\partial A,A) \le \mu(\partial A)\]
    By Theorem 7.4 of \cite{levin2026markov}, we thus have
    \[t_{\rm mix} \ge \frac14 \frac{\mu(A)}{\mu(\partial A)}.\]
\end{proof}

\begin{defin}\label{definrangeRmarkov}
    A ball of radius $R$ is a set $B(R,v)$ where $v$ is a fixed lattice vertex, and $B(R,v)$ is the set of edges which lie along a lattice path of length $\le R$ starting at $v$.

    A Markov chain $P$ has range $R$ if for all $\sigma,\sigma' \in \Omega$ with $P(\sigma,\sigma') \neq 0$, there is a $v$ such that $\sigma$ and $\sigma'$ are equal outside $B(R,v)$.
\end{defin}

We need to modify the definition of the bottles $A[C]$ and bottlenecks $\partial A[C]$ from the main text to satisfy the condition \eqref{eqnbottleneckcondition} for a general range $R$ Markov chain. For this, we need a lemma:

\begin{lemma}\label{lemmahomology}\textbf{Homological stability lemma:}
    For every fixed $R$ there are $\varepsilon > 0$, $M$, and $L_0$ such that for all configurations $S$ and $S'$ of Wegner gauge theory on a 3-torus of side length $L > L_0$, which are equal outside of a ball of radius $R$, the following are true:
    \begin{enumerate}
        \item If each component of $\partial S$ has length $< \varepsilon L$, then all but at most $M$ components of $\partial S'$ have length $< \varepsilon L$.
        \item The combined length of these ``long loops'' is less than $L/2$.
        \item The decoded homology classes are equal: $[C(S)] = [C(S')]$.
    \end{enumerate}
\end{lemma}
\begin{proof}
    There is a maximum number $a_0R^3$ of disjoint dual lattice loops (which occupy plaquettes) which neighbor an edge in $B$. We can take $M = a_0 R^3$, then 1 will automatically be satisfied, because only these loops may differ between $\partial S$ and $\partial S'$.
    
    Furthermore, if all loops in $\partial S$ have length $< \varepsilon L$, the total length of such loops neighboring an edge in $B$ must be $\le a_0 R^3 \varepsilon L$. The most their length can be after an update is therefore
    \[\le a_0 R^3 \varepsilon L + b_0 R^3\]
    for another constant $b_0$. For each fixed $R$, we can thus choose $\varepsilon$ independent of $L$ so that for $L > 3b_0 R^3$, this is less than $L/2$.

    For stability of the decoded homology class, let $\gamma_i$, $i=1,\ldots, k \le M$ denote the set of components of $\partial S$ neighboring an edge in $B$, and let $\gamma_i'$, $i=1,\ldots,k' \le M$ denote the set of components of $\partial S'$ neighboring an edge in $B$. All other components of $\partial S$ and $\partial S'$ are the same, so the difference of the decoded 2-cycles is
    \[C(S)+C(S')= S + S'+ \sum_{i=1}^k D(\gamma_i) + \sum_{j=1}^{k'} D(\gamma_j').\]
    $S+S'$ is supported in $B$ because that's where the update occurs. Furthermore, each $D(\gamma_i)$ or $D(\gamma_j')$ is supported in a bounding cube of side length $<L/4$, and all of these cubes intersect the ball of radius $R$. Their union is thus contained in a cube of side length $L/2 + 2R+2$, where the $+2$ accounts for the shift between the lattice and its dual. This is contractible as long as $L > 4R+4$. So let us choose
    \[L_0 = \max(4R+4,3b_0R^3)\]
    In this case, $C(S) + C(S')$ must be a boundary, so point 3 is satisfied.
\end{proof}

With this lemma in mind, we define, for each $R$, $A[C]$ to be the set of configurations $S$ with $\partial S$ having all components of length $< \varepsilon L$ (call these ``short''), and $S$ having decoded homology class $[C]$, and $\partial A[C]$ to be the set of $S$ where all but $k \in [1,M]$ components of $\partial S$ are short, and the rest have length at least $\varepsilon L$, but with total length less than $L/2$, and having decoded homology class $[C]$. By Lemma \ref{lemmahomology}, any range $R$ Markov chain must pass through $\partial A[C]$ to leave $A[C]$.

\begin{thm}\label{thmWegnermemoryapp} \textbf{Stability of Self-Correcting Memory:}
    Let $\Lambda^L$ denote the $L \times L \times L$ cubic lattice with periodic boundary conditions, and $H_0^L$ the Wegner gauge theory on this lattice. Let $\beta = 1/T$. For all
    \[0 < T < 1/\log 5 \qquad \eta > \log 6 + 1\]
    there exists $C_{\rm stab}(T,\eta) > 0$ with the following property. For every $R$ there exist constants
    \[A_0(T,\eta,R) > 0, \qquad \lambda(T,\eta,R) > 0, \qquad L_0(T,\eta,R)\]
    such that for every $L \ge L_0$, every interaction $V$ on $\Lambda^L$ satisfying
    \[\beta |V|_\eta < C_{\rm stab}(T,\eta)\]
    and every range $R$ Markov chain on $\Lambda^L$ which is $\mu$-reversible, where $\mu \propto \exp(-\beta(H_0^L+V))$, its mixing time satisfies
    \[t_{\rm mix} > A_0 e^{\lambda L}\]
\end{thm}
\begin{proof}
For fixed $R$, $L_0$ and $\varepsilon$ are defined in Lemma \ref{lemmahomology}, and the bottles $A[C]$ and bottlenecks $\partial A[C]$ below that. Our goal is to show for the perturbed Gibbs measure $\frac{1}{Z} e^{-\beta(H_0 + V)}$ that the bottleneck ratio
\[\mu(\partial A[C])/\mu(A[C])\]
is upperbounded by $e^{-\Omega(L)}$ as $L \to \infty$. By Lemma \ref{lemmabottleneck}, this will show an exponentially long mixing time for any Markov dynamics satisfying detailed balance for $\mu$.

Since $A[C]$ and $\partial A[C]$ are $G$-symmetric sets, we have
\[\mu(\partial A[C])/\mu(A[C]) = \mu^G(\partial A[C])/\mu^G(A[C]),\]
where $\mu^G$ is the gauge averaged Gibbs measure $\frac{1}{Z} e^{-\beta(H_0 + V_{\rm eff})}$. Thus it will suffice to upperbound the bottleneck ratio on the right-hand side.

As in the main text, we have an injective map
\[i:\partial A[C] \hookrightarrow A[C] \times \Theta\]
where $\Theta$ is the set of size-$k \in [1,M]$ sets of loops of length in $[\epsilon L,L/2)$ whose total length is less than $L/2$. In particular, for $\Gamma \in \Theta$ define
\[D(\Gamma) = \sum_{\gamma \in \Gamma} D(\gamma).\]
Then $i(S) = (S + D(\Gamma),\Gamma)$. Also set
\[|\Gamma| = \sum_{\gamma \in \Gamma} |\gamma|.\]
We have
\[\label{eqnpeierlsstepone}\mu^G(\partial A[C]) = \frac{1}{Z} \sum_{(S',\Gamma) \in \text{Im}(i)} e^{-\beta(|\partial S'|+|\Gamma|+V_{\rm eff}(S'+D(\Gamma)))} \\
\le \frac{1}{Z} \sum_{S' \in A[C]} e^{-\beta(|\partial S'| + V_{\rm eff}(S'))} \sum_{\Gamma \in \Theta} e^{-\beta(|\Gamma| +  V_{\rm eff}(S'+D(\Gamma))-V_{\rm eff}(S'))}\]
In order to run the Peierls bound, we desire an upper bound on the size of
\[\delta_\Gamma V_{\rm eff}(S) :=V_{\rm eff}(S+D(\Gamma)) - V_{\rm eff}(S)\]
For each $X \subset \Lambda$, $V_{{\rm eff},X}$ in \eqref{eqnlocalinteractioneffective} is gauge invariant. So for it to contribute to $\delta_\Gamma V_{\rm eff}(S)$, it must contain such a lattice 1-cycle $C'$ linking $\Gamma$. In this case, its contribution is upper bounded by $2||V_{{\rm eff},X}||_\infty$. Therefore,
\[|\delta_\Gamma V_{\rm eff}(S)| \le 2\Sum_{X \subset \Lambda \\ \exists C' \subset X , C'\text{ linking }\Gamma} ||V_{{\rm eff},X}||_\infty \]
The linking cycle $C'$ must either be contractible or have length at least $L$ to wrap the 3-torus. Let us therefore separate the sum into ``small'' sets $X$ with $l(X) < L$ and ``large'' sets $X$ with $l(X) \ge L$. The contribution from large sets may be upper bounded by
\[2\sum_{g_0 \in G_0} \Sum_{X:g_0 \in G_0(X) \\ l(X) \ge L} ||V_{{\rm eff},X}||_\infty \le 2L^3 \sup_{g_0 \in G_0} \Sum_{X:g_0 \in G_0(X) \\ l(X) \ge L} ||V_{{\rm eff},X}||_\infty e^{m(l(X)-L)} \\
\le 2L^3 e^{-mL} |V_{\rm eff}|_m =: R_m\]
As long as $|V_{\rm eff}|_m < \infty$, which we will guarantee below, $R_m$ will be an exponentially small remainder which does not contribute to the large $L$ Peierls estimate.

Now we proceed to bound the contribution from small sets $X$ containing a contractible loop $C'$ linking $\Gamma$. Suppose $C'$ has length $r$. Then, since $C'$ is contractible, it is contained in a bounding cube of side length $r/2$. Since it links $\Gamma$, $\Gamma$ must pass through this cube as well. Thus, there is a lattice path $p$ of length $l(p) \le (3/2)r$ from a vertex of $C'$ to a vertex at a corner of a plaquette of $\Gamma$. Let $N(\Gamma,(3/2)r) \subset G_0$ denote the set of vertices which can be connected to $\Gamma$ by such a path. By the above, we have $G_0(C') \cap N(\Gamma,(3/2)r) \neq \varnothing$. Thus, we can give a bound
\[\Sum_{X \subset \Lambda \\ \exists C' \subset X \\ C'\text{ linking }\Gamma, \text{ contractible}} ||V_{{\rm eff},X}||_\infty \le \sum_{r=4}^{L-1} \sum_{g_0 \in N(\Gamma,(3/2)r)} \Sum_{X: g_0 \in G_0(X) \\
l(X) \ge r} ||V_{{\rm eff},X}||_\infty \\
\le |V_{\rm eff}|_m\sum_{r=4}^{L-1} |N(\Gamma,(3/2)r)| e^{-mr}\]
$|N(\Gamma,(3/2)r)|$ is the number of points in an $\ell^1$-metric tube around $\Gamma$. Thus there is a universal constant $c_0$ such that
\[|N(\Gamma,(3/2)r)| \le c_0 r^3 |\Gamma|.\]
Let
\[C_{\rm perim}(m) = 2c_0 \sum_{r=4}^{\infty} r^3 e^{-mr}\]
This is finite and goes to 0 as $m \to \infty$. We have
\[|\delta_\Gamma V_{\rm eff}(S)| \le C_{\rm perim}(m) |V_{\rm eff}|_m|\Gamma| + R_m\]

Combined with \eqref{eqnpeierlsstepone} we obtain
\[\mu^G(\partial A[C]) \le \frac{e^{\beta R_m}}{Z} \sum_{S' \in A[C]} e^{- \beta (|\partial S'|+V_{\rm eff}(S'))}\sum_{\Gamma \in \Theta} e^{-\beta(1 -  C_{\rm perim}(m)|V_{\rm eff}|_m) |\Gamma|} \\
= e^{\beta R_m} \mu(A[C])\sum_{\Gamma \in \Theta} e^{-\beta(1 -  C_{\rm perim}(m)|V_{\rm eff}|_m) |\Gamma|}.\]
We see at this stage that $V_{\rm eff}$ has just served to renormalize the tension of the long loops $\Gamma$. We need to ensure that the tension is still within the radius of convergence of the Peierls estimate. Let
\[\label{eqnqparam} l_* = \lceil \varepsilon L\rceil \qquad q = 5 e^{-\beta(1-C_{\rm perim}(m)|V_{\rm eff}|_m)}.\]
Let $\Theta_k$ be the subset of $\Theta$ consisting of those $\Gamma$ with $k$ components. We have, as long as $q < 1$,
\[\sum_{\Gamma \in \Theta} e^{-\beta( +  C_{\rm perim}(m)|V_{\rm eff}|_m) |\Gamma|} = \sum_{k=1}^M \sum_{\Gamma \in \Theta_k} e^{-\beta( +  C_{\rm perim}(m)|V_{\rm eff}|_m) |\Gamma|} \\
\le \sum_{k=1}^M \left(L^3\sum_{r=l_*}^\infty q^r\right)^k
\\ = \sum_{k=1}^M \left(\frac{L^3 q^{l_*}}{1-q} \right)^k\]
For all $q < 1$, this function has exponential decay $O(\exp\left(-\lambda L \right))$ for all
\[0<\lambda < \varepsilon (-\log q).\]
We will choose $C_{\rm stab}(T,\eta)$ so that for all $V$, $q$ will be bounded away from 1 uniformly in $L$, it will suffice for the theorem to choose $\lambda = \varepsilon (-\log q)/2$ and $A_0$ an appropriately large constant. Indeed, this gives an upper bound on the bottleneck ratio $\mu^G(\partial A[C])/\mu^G(A[C])$ which goes to zero exponentially quickly with $L$. Since on $T^3$, there are 8 possible $[C]$, and all the $A[C]$ are disjoint, one of them will have $\mu^G(A[C]) < 1/2$. Thus, by Lemma \ref{lemmabottleneck}, we will get the desired bound on the mixing time.

To finish the proof, suppose that $T < 1/\log 5$, i.e.
\[\beta > \log 5.\]
Then for all $\eta > \log 6 + 1$, we can choose $m$ and $b$ such that
\[
0 < m < \eta - \log 6 -1 \\
0 < b < \min\left(\eta-m - \log 6 - 1, \frac{\beta -\log 5}{C_{\rm perim}(m)}\right)\]
By the Gauge Averaging Theorem (Theorem \ref{apptheoremstatement}) (here $\Delta = 6$), so long as
\[\beta |V|_\eta < C_{\rm gauge}(\eta,m,b),\]
we obtain a gauge-invariant effective interaction $V_{\rm eff}$ with $\beta|V_{\rm eff}|_m < b$. In particular, by the condition on $b$, $q$ in \eqref{eqnqparam} will be bounded uniformly in $L$ by
\[q < 5e^{-\beta+b C_{\rm perim}(m)} < 1\]
Thus, the theorem is proved with $C_{\rm stab}(T,\eta) = C_{\rm gauge}(\eta,m,b)$.

\end{proof}

\end{document}